\documentclass{article}
\usepackage{enumitem}
\usepackage[english]{babel}
\usepackage{algpseudocodex}
\usepackage{amsfonts}
\usepackage{algorithm}
\usepackage{tikz}
\usepackage{booktabs}
\usepackage{caption}
\usepackage{subcaption}

\newcommand{\opt}{\mathsf{opt}}
\usepackage[letterpaper,top=2cm,bottom=2cm,left=3cm,right=3cm,marginparwidth=1.75cm]{geometry}
\makeatletter
\algnewcommand\algorithmicwhen{\textbf{when}}%
\algdef{SE}[WHEN]{When}{EndWhen}[1]{\algpx@startIndent\algpx@startCodeCommand\algorithmicwhen\ #1\ \algorithmicdo}{\algpx@endIndent\algpx@startCodeCommand\algorithmicend\ \algorithmicwhen}%

\ifbool{algpx@noEnd}{%
  \algtext*{EndWhen}%
  \algtext*{EndCase}%
  \apptocmd{\EndWhen}{\algpx@endIndent}{}{}%
}{}%

\pretocmd{\When}{\algpx@endCodeCommand}{}{}

\ifbool{algpx@noEnd}{%
  \pretocmd{\EndWhen}{\algpx@endCodeCommand[1]}{}{}%
}{%
  \pretocmd{\EndWhen}{\algpx@endCodeCommand[0]}{}{}%
}%
\makeatother
\usepackage{amsmath}
\usepackage{amsthm}

\newtheoremstyle{myplain}
  {\topsep}   
  {\topsep}   
  {\itshape}  
  {}          
  {\bfseries} 
  {.}         
  { }         
  {}          

\theoremstyle{myplain}

\usepackage{graphicx}
\newtheorem{claim}{Claim}
\newtheorem{lemma}{Lemma}
\newtheorem{theorem}{Theorem}
\newtheorem{definition}{Definition}

\usepackage[colorlinks=true, allcolors=blue]{hyperref}
\usepackage{cleveref}
\newcommand{\wpap}{\textsc{WPAP}}
\newcommand{\online}{\textsc{ALG}}
\renewcommand{\opt}{\mathsf{OPT}}
\newcommand{\excess}{\mathsf{excess}}
\newcommand{\layer}{\mathsf{layer}}
\newcommand{\cost}{\mathsf{cost}}
\newcommand{\dist}{\mathsf{dist}}
\newcommand{\DP}{\mathsf{DP}}
\newcommand{\alg}{\mathsf{ALG}}

\newcommand{\ice}{\mathsf{ICE}}

\renewcommand{\wpap}{\mathsf{WPAP}}

\renewcommand{\S}{\mathcal{S}}
\renewcommand{\L}{\mathcal{L}}

\usepackage[parfill]{parskip}

\usepackage{comment}
\def\FULL{} 
\ifdefined\FULL
\excludecomment{shortversion}
\includecomment{fullversion}
\else
\includecomment{shortversion}
\excludecomment{fullversion}
\fi

\author{Afrouz Jabal Ameli\\
	Utrecht University\\
	{\small \texttt{a.jabalameli@uu.nl}} \and 
    Laura Sanit\`a\\Bocconi University\\
    {\small \texttt{laura.sanita@unibocconi.it}} \and
    Moritz Venzin\\Bocconi University\\
    {\small \texttt{moritz.venzin@unibocconi.it}}
 }

\title{Learning-Augmented Online Covering Problems}
\begin{document}
\maketitle
\begin{abstract}
    \noindent We give a very general and simple framework to incorporate predictions on requests for online covering problems in a rigorous and black-box manner. Our framework turns any online algorithm with competitive ratio $\rho(k, \cdot)$ depending on $k$, the number of arriving requests, into an algorithm with competitive ratio of $\rho(\eta, \cdot)$, where $\eta$ is the prediction error. With accurate enough prediction, the resulting competitive ratio breaks through the corresponding worst-case online lower bounds, and smoothly degrades as the prediction error grows. This framework directly applies to a wide range of well-studied online covering problems such as facility location, Steiner problems, set cover, parking permit, etc., and yields improved and novel bounds.  
\end{abstract}

\section{Introduction}
In recent years, there has been significant work to incorporate \emph{ML-advice} in the design of \emph{online algorithms}. The basic premise is to provide additional information to an online algorithm to improve its guarantees. Ideally, whenever the prediction is correct, we expect to break through worst-case lower bounds, and the guarantees should degrade smoothly with the prediction error. This paradigm has been applied (quite successfully) to a wide array of online problems, ranging from caching~\cite{LykourisVassilvitskii21, AntoniadisCoesterEliasPolakSimon23}, scheduling~\cite{KumarPurohitSvitkina18, LattanziLavastidaMoseleyVassilvitskii20}, graph algorithms~\cite{MoseleyXu, AzarPanigrahiTouitouOnlineGraphAlgo} and matching/secretary problems~\cite{DuettingLattanziPaesVassilvitskii21}, to name a few. The rapid evolving of this research direction, often referred to as \emph{learning-augmented algorithms}, is evidenced by the collection of papers in the dedicated website~\cite{Algo_with_pred_website}. Given the attention that the field has gained in the past years, designing general approaches for learning-augmented algorithms is of crucial interest for the community. In this paper, we contribute significantly to this goal by providing a new framework which applies to a general class of online \emph{covering} problems.

To formally introduce our results, we begin by defining the \emph{competitive ratio}. An online algorithm $\alg$ for a given minimization problem is $\rho$-competitive, if for any input $I$, the cost of the online algorithm is at most $\rho$ times the best solution for $I$, that is $\cost(\alg(I)) \leq \rho\cdot \cost(\opt(I))$. 
As mentioned, our framework will apply to online covering problems: Generally speaking, in these problems we are given a sequence of requests $X$ that appear online one by one, and a weighted solution space $\S$. We have to maintain in each step a solution to accommodate them. The term `covering' refers to the following condition, which we assume to hold: 
 For feasible solutions $\S_1$ to $X_1$ and $\S_2$ for $X_2$, $\S_1 \cup \S_2$ is feasible for $X_1 \cup X_2$. 

 In our learning-augmented setting, we have a set \smash{$\hat{X}$} of \emph{predicted} requests. 
 The \emph{prediction error} $\eta$ is then defined as the symmetric difference between 
 \smash{$\hat{X}$} and the set of actually arriving requests $X$, capped at $|X|$. Formally,
 \smash{$\eta := \min(|X|, |X \triangle\hat{X}|)$}, where  \smash{$X \triangle\hat{X} := X\setminus\hat{X} \cup \hat{X}\setminus X$} (this is sometimes referred to as the $\ell_1$ error). 
In the following, we let $k := |X|$.

Our first main result shows that any online covering algorithm with competitive ratio $\rho(k, \cdot)$ depending on $k$, the number of arriving requests, can be turned into a learning-augmented, $O(\rho(\eta, \cdot))$-competitive online algorithm.

\begin{theorem}[Existence]\label{thm:exp_variant}
    Let $\alg$ be an algorithm for an online covering problem with a competitive ratio $\rho(k, \cdot)$, where $k$ is the number of arriving requests. Then, given a prediction \smash{$\hat{X}$}, there exists an online algorithm with a competitive ratio of $O(\rho(\eta, \cdot))$.
\end{theorem}

Our result demonstrates that for any online covering problem, it is always possible to incorporate predictions and obtain a refined competitive ratio that depends on the prediction error~$\eta$. The competitive ratio provided by our framework is \emph{robust}, it is never worse than that of the underlying online algorithm ($\eta \leq k$ by definition), but is significantly better when the predictions are accurate.

Our framework is based on a charging scheme that relies on a subtle \emph{decomposition} of the solution space. Given some prediction \smash{$\hat{X}$}, we compute \emph{offline} partial solutions \smash{$\S_1, \S_2, \ldots \subseteq \S$}, such that their union is feasible for \smash{$\hat{X}$}. In the online phase, we use a $\rho(k, \cdot)$-competitive algorithm to satisfy the requests in a black-box fashion. We only keep track of its expenses: whenever its cost exceeds that of some partial solution $\S_i$, we buy $\S_i$ (in addition to whatever the online algorithm buys). Interestingly, this simple framework turns out to be powerful enough to guarantee a competitive ratio that depends on the misprediction error, rather than the actual request sequence.

We then turn into discussing \emph{efficiency}. In fact, 
whenever the decomposition mentioned above can be computed efficiently, then the overall online algorithm turns into a polynomial time algorithm. However, while such a decomposition always exists (yielding the existential result), it is not always possible to compute it efficiently. This is not surprising, as there are covering problems for which a $O(\rho(\eta, \cdot))$-competitive algorithm \emph{cannot} be achieved in polynomial-time, assuming $\mathsf{P}\neq\mathsf{NP}$. An example of this is the classical Set Cover problem (see Theorem~\ref{thm:LBSetCover}). Nevertheless, we are still able to provide a \emph{relaxed} notion
of decomposition, that we call an \emph{$(\alpha, \gamma)$-decomposition} (details in Section~\ref{sec:framework_formal}). This relaxed notion allows us to smoothly `compromise' between the competitive ratio and the efficiency of the overall framework. The resulting guarantees are resumed in the following theorem.

\begin{theorem}[Efficiency]\label{thm:poly_variant}
    Let $\alg$ be an algorithm for an online covering problem
    with a competitive ratio $\rho(k, \cdot)$, where $k$ is the number of arriving requests.
    Assume that an $(\alpha, \gamma)$-decomposition for the \emph{offline} version of the covering problem can be computed in polynomial-time. Then, given a prediction \smash{$\hat{X}$}, there exists a polynomial-time online algorithm with a competitive ratio of 
    $$O(\alpha)\cdot \rho(\eta, \cdot) + O(\alpha\cdot \log^{-1}(\tfrac{\gamma}{\gamma-1})\cdot \log k).$$
    In particular, if $\gamma=1$, the competitive ratio becomes $O(\alpha \cdot \rho(\eta, \cdot)).$ 
\end{theorem}
Note that a $(1,1)$-decomposition always exists and yields Theorem~\ref{thm:exp_variant}.
This \emph{additive} term bypasses the aforementioned impossibility to turn any $\rho(k, \cdot)$-competitive algorithm into a learning-augmented, $O(\rho(\eta, \cdot))$-competitive algorithm in polynomial time.

We demonstrate the versatility of our approach, by applying it to the vast class of online covering problems. In particular, our framework yields
the \emph{first} learning-augmented algorithms when considering prediction on requests, for classical problems like
set cover, non metric facility location, path augmentation, and hard generalizations of Steiner problems.

We conclude this introduction by giving an account of example applications (Section~\ref{sec:applications}), and a comparison to related work (Section~\ref{sec:comparison_to_related_work}). Our framework is formally presented in Section~\ref{sec:framework_formal}, and the details for our example applications are discussed in Section~\ref{sec:applications_formal}. In Section~\ref{sec:experimental_evaluation} we present the experimental evaluation of our approach for the online set cover problem. 
\begin{shortversion}
All missing details and proofs can be found in the full version attached at the end.
\end{shortversion}

\subsection{Applications}\label{sec:applications}
\textbf{Set Cover and Non-Metric Facility Location } An instance $(X', \S, \cost)$ of online set cover is given by a ground set $X'$, an ensemble of sets \smash{$\S \subseteq 2^{X'}$}, as well as a cost function $\cost: S \rightarrow \mathbb{R}_{\geq 0}$. A subset of the elements $e\in X \subseteq X'$ arrive one-by-one. This is a special case of the non-metric facility location problem. There, an instance is given by a client-facility graph $(C', F, \cost: F \rightarrow \mathbb{R}_{\geq 0}, \dist: C' \times F \rightarrow \mathbb{R}_{\geq 0} \cup \{+\infty\})$. Clients $c\in C\subseteq C'$ are released one-by-one. For both problems, there exists a \emph{randomized} online algorithm with competitive ratio $O(\log k \log |F|)$ (resp. $O(\log k \log |\S|)$),~\cite{AlonAwerbuchBuchbinderNaor}. As we discuss in Section~\ref{sec:set_cover_decomp}, given some predicted requests corresponding to clients / elements, an $(O(1), 1-1/e)$-decomposition can be computed in polynomial time. Hence, our framework yields a $O(\log \eta \log |\S|)$- (resp. $O(\log\eta\log|\S|)$-) competitive algorithm, plus an additive $O(\log k)\cdot \opt$, if it is required to run in polynomial time. 

\textbf{Parking Permit and Weighted Path Augmentation} An instance of weighted path augmentation ($\wpap$) is given by $(E, \L, \cost)$. $E$ is a set of elements, that correspond to edges of a path $P$, indexed from $1$ to $n$, and $\L$ is a set of \emph{links}, weighted according to a $\cost$ function. Each link covers a consecutive interval of elements (subpath of $P$). In the online phase, elements arrive one-by-one, and we need to immediately and irrevocably select a link covering it. This can be thought as connectivity augmentation problem where edges of a path can fail online, and we need to restore connectivity. The Parking Permit problem is a special case of $\wpap$. For any element (edge of the path), we can select $K$ different permits (links) of fixed durations (i.e.\,length). For both problems, it is possible to devise a deterministic algorithm that is $O(\log k)$-competitive, as well as a randomized online algorithm that is $O(\log \log k)$-competitive against oblivious adversaries. Since a $(1,1)$-decomposition can be computed in polynomial time, this results in a deterministic, $O(\log\eta)$-competitive online algorithm and a randomized, $O(\log\log\eta)$-competitive online algorithm. We provide the details in Section~\ref{sec:wpap}.

\textbf{Metric facility Location and Steiner Tree variants } Metric facility location is defined analogously to the non-metric facility location problem (NMFL), except that both clients and facilities are vertices of a \emph{metric} graph $G = (V, E)$. This metric assumption makes a big difference, the best competitive ratio is \smash{$O(\tfrac{\log k}{\log \log k})$},~\cite{Meyerson_online_facility_location, Fotakis_facility_location}. Our framework yields a polynomial-time, \smash{$O(\tfrac{\log \eta}{\log \log \eta})$}-competitive learning-augmented algorithm. This should be compared with the $O(\log\triangle)$-competitive ratio from~\cite{AzarPanigrahiTouitouOnlineGraphAlgo}, although their results holds with respect to a slightly different error measure. We elaborate in Section~\ref{sec:comparison_to_related_work}. 
In Steiner tree (resp. forest) problems, given a metric graph $G$, terminals $\{t_1, t_2, \ldots, t_k\}$ show up one-by-one (resp. in pairs), and one needs to maintain a connected subgraph connecting all arrived terminals (terminal pairs). In the edge-weighted setting, there are $O(\log \eta)$-competitive algorithms for these two problem~\cite{MoseleyXu, AzarPanigrahiTouitouOnlineGraphAlgo}, which we can also obtain with our framework. However, for several important variations such as group Steiner tree / forest, as well as the node-weighted and directed setting, no learning-augmented algorithms are known. For these problems, the competitive ratio is of order \smash{$O(\log^{O(1)} k \log^{O(1)} |V|)$},~\cite{NaorPanigrahiSingh, ChakrabartyEneKrishnaswamyPanigrahi}. Our framework directly improves this to \smash{$O(\log^{O(1)} \eta \log^{O(1)} |V|)$} in the learning-augmented setting. We provide the details in Section~\ref{sec:steiner}. 

\subsection{Comparison to related work and techniques}\label{sec:comparison_to_related_work}
The study of general schemes in the context of learning-augmented online algorithms for set cover (and other covering problems) was initiated by~\cite{BamasMaggioriSvensson20_primaldual}. They devise a primal-dual approach for covering problems that incorporates predictions. 
Subsequent works on the problem then consider further variations, see for instance~\cite{AnandGeKumarPanigrahi, GuptaPanigrahiSubercaseauxSun_epsilon22, AzarPanigrahiTouitouDiscreteSmoothness23}. 
However, a common feature of these works, is that the prediction is on the sets, i.e.\,the online algorithms receives hints on what an optimal solution (conditioned on the set of actually arriving elements) looks like. As such, these results are incomparable to ours.
Despite this, we would like to mention that the results in~\cite{BamasMaggioriSvensson20_primaldual, GrigorescuEtAl2022, AzarPanigrahiTouitouDiscreteSmoothness23} can still be \emph{combined} with our framework. Informally, given a prediction on sets in an optimal solution, they provide an algorithm that maintains a $f(\mathsf{Qual}_\mathsf{pred})$-competitive \emph{fractional} solution. $\mathsf{Qual}_\mathsf{pred}$ is the \emph{quality} of the prediction (with respect to some error metric), and $f(\mathsf{Qual}_\mathsf{pred})\leq O(\log |\S|)$. To obtain an integral solution, they round it in an online fashion. This incurs a multiplicative loss of $O(\log k)$. Hence, combined with our framework, this yields a polynomial-time $O(\log \eta \cdot f(\mathsf{Qual}_\mathsf{pred}) +\log k)$-competitive algorithm for online set cover that incorporates predictions on arriving requests, as well as sets of an optimal solution.

A general framework that is close to ours is the one by~\cite{AzarPanigrahiTouitouOnlineGraphAlgo}. Here they study an interesting generalization of the error we consider, specialized to problems on metric graphs. 
Specifically, they consider the \emph{metric error with outliers} given by $(\triangle, D)$: for some \smash{$\hat{T} \subseteq \hat{X}$}, $T \subseteq X$ with \smash{$|T| = |\hat{T}|$}, and a min-cost matching between $T$ and \smash{$\hat{T}$} of cost $D$, and \smash{$|\hat{X}\setminus\hat{T} \cup X\setminus T| = \triangle$}. An important feature of their work, is that their guarantees hold with respect to the best choice of $\triangle$ and $D$. In particular, setting \smash{$\hat{T} = T = \hat{X} \cap X$} (consequently $D = 0$), their error measure is always smaller than $\eta$ for metric graph problems. On a high-level, the overall approach is similar as they also compute partial solutions to the previously arrived elements, and pay towards predicted requests. This is decided based on an instance of a \emph{prize-collecting} variant of the problem. However, to obtain their guarantees, they require the underlying online algorithm to be \emph{subset-competitive}, a non-trivial property (for instance for some online algorithms for the metric facility location problem). In contrast, we rely on a \emph{more involved} decomposition, but our results \emph{do not} depend on specifics of the underlying online algorithm and can be applied in a completely modular and black-box way. As a direct consequence, we can also apply our results to any standard algorithm for these problems. In particular, by leveraging the \emph{best} known online algorithm for metric facility location with competitive ratio \smash{$O(\tfrac{\log k}{\log\log k})$}, we obtain a bound of \smash{$O(\tfrac{\log\eta}{\log\log\eta})$}. In comparison, they give a $O(\log \triangle)$-competitive ratio (plus an additive term of $D$), as they rely on a $O(\log k)$-competitive algorithm to ensure subset-competitiveness. 
As a side note, we remark that there has significant work on learning-augmented online facility location, see~\cite{AnandGeKumarPanigrahi, Almanza_FL_wPred,jiang2022online, AzarPanigrahiTouitouDiscreteSmoothness23, Fotakis_FL_with_predictions}. However, all these works assume that predictions are hints on the location of optimal facilities. As such, these results are incomparable to ours. The framework in~\cite{AzarPanigrahiTouitouOnlineGraphAlgo} also applies to Steiner tree and Steiner forest, yielding a $O(\log \triangle)$- (plus additive $D$) competitive ratio (a bound of $O(\log \eta)$ for Steiner tree was previously given in~\cite{MoseleyXu}). Our framework recovers this latter guarantee of $O(\log\eta)$, but it also applicable for the generalisations in Section~\ref{sec:applications}.

For the parking permit and its generalisation, the weighted path augmentation problem, there were no known learning-augmented algorithms. The parking permit was introduced in~\cite{Meyerson05ParkingPermit}, where a deterministic $O(\log n)$-, and a randomized, $O(\log\log n)$-competitive algorithm was given. However, it is straightforward to adapt this algorithm to yield a $O(\log k)$- (resp. $O(\log\log k)$-) competitive algorithm for parking permit. A generalisation of $\wpap$ was first studied by~\cite{NaorUmbohWilliamson}, and they provide a deterministic, $O(\log n)$ competitive algorithm for the problem, as well as a fractional, $O(\log\log n)$-competitive algorithm for $\wpap$. We adapt these ideas to yield a deterministic, $O(\log k)$-competitive algorithm for $\wpap$, as well as a randomized, $O(\log\log k)$-competitive algorithm. Our framework leverages these results to yield a learning-augmented, $O(\log\eta)$ (resp. $O(\log\log\eta)$-) competitive algorithm for deterministic (resp. randomized) $\wpap$ and parking permit.

We conclude this subsection with a final remark.
The error $\eta$ we consider is the $\ell_1$ error, the term $\min\{|X|, \cdot\}$ simply follows from robustification. 
\begin{shortversion}
    As such, our error measure is very simple and efficient to learn, it suffices to take the element-wise majority over a  \emph{logarithmic} number of past instances (see full version). It also naturally handles fractional predictions, where each request is associated with some probability or confidence parameter, as these can be rounded while preserving the overall error.
\end{shortversion} 
\begin{fullversion}
    As such, our error measure is very simple and efficient to learn. We briefly sketch this. Consider the setting where we are given a distribution over instances of some covering problem with requests in $X'$. The goal is to find a "best-possible" predictor of the form \smash{$\hat{x} \in \{0, 1\}^{|X'|}$}, which minimizes the expected $\ell_1$-distance (e.g.\,$\eta$) to a request sequence $x\in \{0,1\}^{|X'|}$ sampled from the distribution. Note that by rounding, we can indeed assume $\hat{x}$ to be integral. Since the $\ell_1$-norm is just a sum over all components, it suffices to minimize the expected Hamming distance for each component separately. For any such component, taking the majority over $O(\log |X'|)$ i.i.d. samples, ensures that with probability at least $1-1/\mathsf{poly}(|X'|)$, the resulting $\{0, 1\}$-majority vote is within $0.49$ of the true predictor $\hat{x}$. Taking the union bound over all $|X'|$ ensures correctness for all components (with high probability). For a textbook argument, see e.g.~\cite[Chapter~$4$]{Shalev-Shwartz_Ben-David_2014}. Finally, we note that the error measure $\eta$ also naturally handles fractional predictions, where each request is associated with some probability or confidence parameter. Again, this follows from the fact that we can always round such a prediction (round up each component if it is larger than $1/2$, round down otherwise) while preserving the (expected) error.
\end{fullversion}

\section{Framework for predictions on requests}\label{sec:framework_formal}
We now describe $\ice$ (\textbf{I}teratively \textbf{C}harge \textbf{E}xpenses), a deterministic online procedure to incorporate predictions on requests. This procedure operates in two stages. In the first, offline phase, we receive the prediction \smash{$\hat{X}$} on the requests. We compute a decomposition for \smash{$\hat{X}$}, i.e.\,
$$\hat{X}_1 \sqcup \hat{X}_2 \sqcup \ldots \sqcup \hat{X}_n = \hat{X},$$
and corresponding feasible solutions $\S_1, \S_2, \ldots, \S_n$, of respective costs $c(\S_i)$. In the online phase, we keep track of the cost by $\alg$ incurred on elements in \smash{$X \cap \hat{X}$}. Whenever this exceeds the cost of $\S_i$, we buy $\S_i$. We resume this in Algorithm~\ref{alg:ice}.
\begin{algorithm}
\begin{algorithmic}[1]
  \State \textbf{Input:} Online algorithm $\alg$ with competitive ratio $\rho(k,\cdot)$ and predicted requests $\hat{X}$.
  \State \textbf{Offline phase:} Compute  \smash{$(X_i, \S_i)$}s as detailed in Section~\ref{sec:offline_phase}.\label{Line:decomposition}
  \State Initialize two instances of $\alg$, $\alg_+$ and $\alg_-$.
  \State Set $\excess \leftarrow 0$, $\layer \leftarrow 1$.

  \When{request $x\in X$ arrives}
  \If{$x \notin \hat{X}$}
  \State Pass $x$ to $\alg_+$.
  \Else
  \State Pass $x$ to $\alg_{-}$. Increase $\excess$ by the cost of the partial solution $\alg_-$ selects for $x$.
  \While{$\excess\ge c(\S_{\layer})$}
  \State Buy $S_\layer$, $\layer \leftarrow \layer+1$, $\excess \leftarrow \excess - c(\S_\layer)$ and reinitialize $\alg_-$.
  \EndWhile
  \EndIf
  \EndWhen

\end{algorithmic}
\caption{$\ice$ - Iteratively Charge Expenses}\label{alg:ice}
\end{algorithm}

The running time of our procedure is at most that of the online algorithm that it uses in a black-box manner, as well as the time to compute the decomposition in Line~\ref{Line:decomposition}. We formalize this in the following subsection.

\subsection{Offline phase of the General Framework}\label{sec:offline_phase}
Before describing the offline computation,
we need the following definition.

\begin{definition}[Efficient $(\alpha, \gamma)$-decomposable]\label{def:decomposition} An offline covering problem with request set $X$ is said to be $(\alpha, \gamma)$-decomposable if for all $i\in \{1, 2, \ldots, |X|\}$, we can compute a set $B \subseteq X$ and a feasible solution $\S_B$ that services the requests in $B$ where 
\begin{itemize}
    \item The cost of $\S_B$ is at most $\alpha$ times the minimum cost of servicing~$i$ requests in $X$, i.e.\,
    $$ c(\S_B) \leq \alpha \cdot \min_{A \subseteq X, |A| = i} \{c(\S_A) \mid \S_A \text{ feasible for } A\}, \text{ and, }$$
    \item $|B| \geq i/\gamma.$
\end{itemize}
Whenever $B$ and $\S_B$ can be computed in polynomial time, we say the problem is \emph{efficiently} $(\alpha, \gamma)$-decomposable.
\end{definition}
 Note that any problem is $(1, 1)$-decomposable, e.g.\,through integer programming or enumeration techniques. However, this may take super-polynomial time, as is the case, e.g., for set cover. 

Next, we use this property of $(\alpha, \gamma)$-decomposability, to obtain an $(\alpha, \gamma)$-\emph{decomposition}.
For convenience, we set $g_{\alpha,\gamma}(t):=\alpha$ if $\gamma=1$ and $g_{\alpha,\gamma}(t):= \alpha\cdot (1+\log^{-1}(\tfrac{\gamma}{\gamma-1})\cdot\log t)$ if $\gamma>1$.

 We are going to split \smash{$\hat{X}$} through an iterative process: starting with \smash{$R_0 := \hat{X}$}, we select a subset $X_1 \subseteq R_0$, update $R_1 \leftarrow R_0 \setminus X_1$, and continue anew with $R_1$. Then for every $i$, we denote by $\S_i$ the ($\alpha$-approximate) solution computed to cover the elements in $X_i$, as in Definition~\ref{def:decomposition}. An  $(\alpha, \gamma)$-decomposition
 is then a pair $(X_i, \mathcal S_i)$, where \smash{$\hat{X}$}$ = X_1\cup X_2 \cup \dots \cup X_r$, and the following properties hold for each $i$:

\begin{itemize}
    \item[(A)] $|X_i| \geq \frac{|\hat{X}| - (|X_1| + \ldots + |X_{i-1}|)}{2} $.
    \item[(B)] If $c(\S_i) < 2 c(\S_{i-1})$, then $8c(\S_{i-1})<c(\S_{i+1})$.
    \item[(C)] If $c(\S_i)>10c(\S_{i-1})$, then $c(\S_i)\le g_{\alpha,\gamma}(|\hat{X}|)\cdot C_1$, where $C_1$ is the minimum cost solution that covers $\lceil\frac{|\hat{X}| - (|X_1| + \ldots + |X_{i-1}|)}{2}\rceil$ elements from $R_{i-1}$. 
    \item[(D)]  \smash{$c(\S_i)\le g_{\alpha,\gamma}(|\hat{X}|)\cdot C_2$}, where $C_2$ is defined as follows: for $i \geq 2$, $C_2$ is the minimum cost solution that covers $|X_i|$ elements from $R_{i-1}$; for $i=1$, $C_2$ is the minimum cost solution that covers \smash{$\lceil |\hat{X}|/ 2\rceil$} elements from $R_{0}$.
    
\end{itemize}

The main result of this section is the following.

\begin{theorem}\label{thm:DecompostionConstruction}
    There is a polynomial-time construction guaranteeing properties (A), (B), (C), and~(D), provided that the problem is efficiently $(\alpha, \gamma)$-decomposable.
\end{theorem}

\begin{fullversion}
To prove it we first need the following useful lemma.
\begin{lemma}\label{lem:DecompositionToApprox}
    
If the problem is efficiently $(\alpha,\gamma)$-decomposable, a $g_{\alpha,\gamma}(|\hat{X}|)$-approximation for computing a minimum cost solution to cover at least $i$ elements from $\hat{X}$ can be computed in polynomial time.
\end{lemma}

\begin{proof}
    We iteratively apply the following process. Start with $h=i$, $F:=\hat{X}$ and $S:=\emptyset$. At each step compute a solution $S'$ in polynomial time, that covers at least $\frac{h}{\gamma}$ elements whose cost is at most $\alpha$ times the cost of the minimum cost solution that covers a subset of size $h$ from $F$. $S'$ covers a subset $F'$ of $F$, with $|F'|\ge h/\gamma$. We update $S \leftarrow S\cup S'$, $F\leftarrow F\setminus F'$, and $h\leftarrow h-|F'|$. We stop as soon as $h\le 0$.
    
    When the process terminates $S$ contains a solution that covers at least $i$ elements. The number of steps required is one if $\gamma=1$ and otherwise is upperbounded by $\log_{1-1/\gamma} i+1$ as we covers at least $h/\gamma$ many elements.
    
    Also the solution $S'$ computed at each stage is an $\alpha$-approximation for minimum cost solution to cover at least $i$ elements from $\hat{X}$. Thus, if $\gamma=1$, this is an $\alpha$-approximation and otherwise it is an $(\alpha\cdot (1+\log^{-1}(\tfrac{\gamma}{\gamma-1})\cdot\log i))$-approximation. The claim follows as $i\le |\hat{X}|$.
\end{proof}

We are now going to detail the construction. 
$\S_1$ is a solution that covers a subset $X_1\subseteq R_0$ such that $|X_1|\ge \lceil|R_0|/2\rceil=\lceil|\hat{X}/2|\rceil$ many elements of $R_0=\hat{X}$, such that its cost is at most $g_{\alpha,\gamma}(\hat{X})$ times the cost of the optimal set that covers these many elements from $R_0$. Observe that $X_1$ and $S_1$ can be simply computed by Lemma~\ref{lem:DecompositionToApprox} and hence it satisfies Properties (A) and (D) for $i=1$. 

Now for every $i\ge 1$, we first compute $R_i$, by simply setting $R_i\leftarrow R_{i-1}\setminus X_{i}$. Now we show how we compute $\S_{i+1}$ and $X_{i+1}$.

First, for every $j\ge\lceil\tfrac{|R_i|}{2}\rceil$ we find a subset $X_{i+1,j} \subseteq R_i$ of elements, together with a feasible solution $\S_{i+1,j}$ to cover it, such that:

\begin{itemize}
    \item[i)] $|X_{i+1,j}|\ge j$,
    \item[ii)] $\S_{i+1,j}$ is an $g_{\alpha,\gamma}(| \hat{X}|)$-approximation for the minimum cost to cover among all subsets of $R_i$ of cardinality $j$,
    
    \item[iii)] $c(\S_{i+1,j})\le c(\S_{i+1,j+1})$,
    \item[iv)] If $\S_{i+1,j}$ covers more than $j$ elements of $R_i$ then $c(\S_{i+1,j+1})=c(S_{i+1,j})$,
    \item[v)] and for every $e\in R_i\setminus X_{i+1,j}$, it holds that $c_e+c(\S_{i+1,j})\ge c(\S_{i+1,j+1})$, where $c_e$ is the cost of the cheapest set that covers $e$.

\end{itemize}

For computing $X_{i+1,j}$s and $\S_{i+1,j}$s that satisfy i) to v) we first start with computing $X_{i+1,j}$s and $\S_{i+1,j}$s that satisfy i) and ii), using Lemma~\ref{lem:DecompositionToApprox} by replacing $\hat{X}$ by $R_i$ and replacing $i$ by $j$ in this lemma. Now, if for any $j$, $c(\S_{i+1,j})>c(\S_{i+1,j+1})$, we update $X_{i+1,j}\leftarrow X_{i+1,j+1}$ and $\S_{i+1,j}\leftarrow \S_{i+1,j+1}$. The updated $X_{i+1,j}$s and $\S_{i+1,j}$s satisfy i), ii), and iii).

Now assume that iv) or v) does not hold for some $j$; Let $j^*$ be the smallest $j$ that has this property. If iv) does not hold for $j^*$, then we simply update $\S_{i+1,j^*+1}\leftarrow \S_{i+1,j^*}$ and we update $X_{i+1,j^*+1}$ to the set of all elements in $R_i$ that are covered by $\S_{i+1,j^*}$. Otherwise, if v) does not hold for $j^*$, then we take an element $e\in R_i\setminus X_{i+1,j^*}$, where $c_e$ is minimum and  $c_e+c(\S_{i+1,j})< c(\S_{i+1,j^*+1})$. We update $X_{i+1,j^*+1}\leftarrow X_{i+1,j^*}\cup (S_e\cap R_i)$ and $\S_{i+1,j^*+1}\leftarrow \S_{i+1,j^*} \cup \{S_e\}$, where $S_e$ is the cheapest set that covers $e$.

Note that after these modifications i), ii), and iii) are still satisfied and iv) and v) are still satified for any $j\le j^*$. Also properties iv and v) are now ensured for $j^*$ as well. Hence by repeating this process we can ensure i) to v) for every $j$. Now using $X_{i+1,j}$s and $S_{i+1,j}$s, we use the following rules to construct $X_{i+1}$ and $S_{i+1}$:





\textbf{Case 1}: If $c(\S_{i+1,\lceil|R_i|/2\rceil})\ge 2c(\S_{i})$, we set $X_{i+1}\leftarrow X_{i+1,\lceil|R_i|/2\rceil} $ and $\S_{i+1} \leftarrow \S_{i+1,\lceil|R_i|/2\rceil}$.

\textbf{Case 2}: If not, let $\ell$ be the largest number such that  $c(\S_{i+1,\ell})\le 10 c(\S_{i})$ (Clearly $\ell\ge \lceil\frac{|R_i|}{2}\rceil$). In this case we set $\S_{i+1}\leftarrow \S_{i+1,\ell}$ and $X_{i+1}\leftarrow X_{i+1,\ell}$.


Now we have all the ingredients required to prove Theorem~\ref{thm:DecompostionConstruction}
\begin{proof}[Proof of Theorem~\ref{thm:DecompostionConstruction}]
    It suffices to show that the above construction yields properties (A), (B), (C), and~(D).
    Properties (A) and (D) are clearly achieved by construction and by i) to v). 
    
    Property (C) is also achieved by construction, as if $ 10c(\S_i)<c(\S_{i+1})$ then in our construction we must be in \textbf{Case 1}. Now in \textbf{Case 1}, $\S_{i+1}$ is simply $\S_{i+1,\lceil|R_i|/2\rceil}$, which is a $g_{\alpha,\gamma}(\hat{X})$-approximation for a solution that covers $\lceil|R_i|/2\rceil$ many elements of $R_i$, and hence the claim. 
    
    Now we show that Property (B) holds; To prove this, note that in our construction if $c(\S_i)<2c(\S_{i-1})$ then we should be in \textbf{Case 2} at stage $i$. However, in \textbf{Case 2}, by our choice of $\ell$, the cost to cover any additional element (i.e. any element in $R_{i}$) should be at least $8c(\S_{i-1})$; as otherwise, by v), $\ell$ is not the largest possible number as described in \textbf{Case 2}, a contradiction. Therefore, $c(\S_{i+1})\ge 8c(\S_{i-1})$
\end{proof}

\end{fullversion}

\subsection{General Framework Analysis}

In this section, we analyze the competitive ratio of Algorithm~\ref{alg:ice}. Assume that upon the termination of our algorithm (Algorithm~\ref{alg:ice}), we are at layer $i+1$ and $\overline{\excess}$ is the value of $\excess$ when the algorithm ends. That is, in the previous steps we were able to buy $\S_1 \cup \ldots \cup \S_i$, and $\layer=i+1$ when the algorithm terminates, with $\overline{\excess}<c(\S_{i+1})$. We also define \smash{$\Delta_-:=|\hat X\setminus X|$}. 

We begin our analysis by providing the following two Lemmas that give a lower bound on the size of the optimal solution. 

\begin{lemma}\label{Lem:LinearBoundOnOptGen}
    Let $m$ be a positive integer.
    If $\Delta_- \le |R_m|$, then 
    $\opt \geq \max_{1\le j \le m}\{c(\S_j)\} / g_{\alpha,\gamma}(|\hat{X}|)$.\label{claim:cost}
\end{lemma}

\begin{proof}
    Let $j\in \{1,\ldots,m\}$. Now as $\Delta_- \le |R_m|$, then $\Delta_- \le |R_j|$. This means that at least  $|X_j|$ many elements from $R_{j-1}$ appear in our instance (otherwise $\Delta_->|R_j|$ as $|R_j|+|X_j|=|R_{j-1}|$). So, by Property (D), the cost to cover these elements is at least \smash{$c(\S_j)/g_{\alpha,\gamma}(|\hat{X}|)$}.
\end{proof}

\begin{lemma}\label{Lem:LogarithmicBoundOnOptGen}
    The following holds:
    \begin{itemize}
        \item For every $j\in \{1,\ldots,i\}$, $\opt\in \Omega (c(\S_j)/\rho(\Delta_-,.))$ or $\opt\in \Omega (c(\S_j)/g_{\alpha,\gamma}(|\hat{X}|))$, and
        \item $\opt\in \Omega (\overline\excess/\rho( \Delta_-,\cdot)$ or $\opt\in \Omega (\overline\excess/g_{\alpha,\gamma}(|\hat{X}|)).$ 
    \end{itemize}
\end{lemma}

\begin{shortversion}
\begin{proof}

For every $j \in \{1, \dots, i+1\}$, let us denote by $a_j$ the minimum of the value of $\excess$ at the beginning of layer $j$ and the value of $c(\S_j)$. Also for every $j \in \{1, \dots, i\}$ we define $c(\online_j):=c(\S_j)-a_j$. For $j=i+1$, we define $c(\online_j):=\overline{\excess}-a_j$. Consequently, we have $a_1 = 0$, and $0 \leq a_j \leq c(\mathcal{S}_j)$.

To better explain the above notation, consider the following scenario;  
Assume we are at $\layer=j$, and then an element $x\in \hat{X}$ arrives. According to line 9 of Algorithm~\ref{alg:ice}, ALG buys the partial solution $S$ for $x$. Now we update $\excess\leftarrow\excess+c(S)$. Now if $\excess\ge c(\S_\layer)$, let $z$ be the largest number such that \smash{$\excess\ge \sum_{i=\layer}^z c(\S_i)$}. In this case, we will have \smash{$a_{j+1}=c(\S_{j+1})$, $a_{j+2}=c(\S_{j+2})$}, $\dots$, $a_z=c(\S_z)$, and \smash{$c(\online_{j+1})=c(\online_{j+2})=\dots=c(\online_z)=0$}; Also we will have $a_{z+1}=\excess-\sum_{i=\layer}^z c(\S_i)$.

To prove the claim for any $j$, we will show that $\opt\in \Omega(\min\{\frac{c(\online_j)}{\rho(\Delta_-,\cdot)},\frac{c(\online_j)}{g_{\alpha,\gamma}(|\hat{X}|)}\}+a_j)$. 
    
    We first show that $\opt \ge a_j$. As $a_1=0$, the claim is trivial for $j=1$. If $j\ge2$ the cost that the online algorithm spends at the stage which led to purchasing $\S_{j-1}$ is at most $a_j$, and therefore $a_j$ is a lower bound on $\opt$; This is true as without loss of generality, $\alg$ never pays more than $\opt$ for any element. 

    To complete the argument, we will show $\opt\in \Omega(\frac{c(\online_j)}{\max\{\rho(\Delta_-,\cdot),g_{\alpha,\gamma}(|\hat{X}|)\}})$.  Note that the number of elements handled by our algorithm in layer $j$ is at most $|X_j|$, as otherwise by Property (D), the cost to cover only these elements is at least \smash{$\opt \ge c(\online_j)/g_{\alpha,\gamma}(|\hat{X}|)$}.
    
    So  by the assumption on the competitive ratio of $\alg$ we have \smash{$\opt\in \Omega(\frac{c(\online_j)}{\rho( |X_{j}|,\cdot)})$} (the online algorithm for elements that arrive at layer $j$, is an instance with at most $|X_j|$ elements and hence is $\rho(|X_j|,\cdot)$-competitive). This means that if $\Delta_-\ge  |X_j|/2$ we are done as  $\rho(\Delta_-,\cdot) \in  \Omega(\rho(|X_{j}|,\cdot))$ (note that the competitive ratio can grow at most linearly). So we assume that \smash{$\Delta_-< |X_j|/2$}. As $ X_j\subseteq R_{j-1}$, then $\Delta_-\le |R_{j-1}|$ and by Lemma~\ref{Lem:LinearBoundOnOptGen}, if $j>1$ then we have \smash{$\opt \ge c(\S_{j-1})/g_{\alpha,\gamma}(|\hat{X}|)$}.

    We distinguish two cases; If $j>1$ and $c(\S_{j})\le 10 c(\S_{j-1})$, then $\opt \ge  c(\S_{j-1})/g_{\alpha,\gamma}(|\hat{X}|)$ already implies $\opt \in \Omega(c(\online_j)/g_{\alpha,\gamma}(|\hat{X}|))$, and hence the claim. Now consider the case $j=1$ or $c(\S_{j})> 10 c(\S_{j-1})$; As \smash{$\Delta_-<|X_j|/2$}, then we have \smash{$\Delta_-<|R_{j-1}|/2$}.  By Property (C) for $j>1$ and Property (D) for $j=1$, it holds that \smash{$\opt \ge c(\S_{j})/g_{\alpha,\gamma}(|\hat{X}|)$}.
\end{proof}
\end{shortversion}

\begin{fullversion}
    \begin{proof}

For every $j \in \{1, \dots, i+1\}$, let us denote by $a_j$ the minimum among the value of $\excess$ at the beginning of layer $j$ and $c(\S_j)$. Also for every $j \in \{1, \dots, i\}$ we define $c(\online_j):=c(\S_j)-a_j$. For $j=i+1$, we define $c(\online_j):=\overline{\excess}-a_j$. Consequently, we have $a_1 = 0$, and $0 \leq a_j \leq c(\S_j)$.

To better explain the above notation, consider the following scenario;  
Assume we are at $\layer=j$, and then an element $x\in \hat{X}$ arrives. According to line 9 of Algorithm~\ref{alg:ice}, ALG buys the partial solution $S$ for $x$. Now we update $\excess\leftarrow\excess+c(S)$. Now if $\excess\ge c(\S_\layer)$, let $z$ be the largest number such that \smash{$\excess\ge \sum_{i=\layer}^z c(\S_i)$}. In this case, we will have \smash{$a_{j+1}=c(\S_{j+1})$, $a_{j+2}=c(\S_{j+2})$}, $\dots$, $a_z=c(\S_z)$, and \smash{$c(\online_{j+1})=c(\online_{j+2})=\dots=c(\online_z)=0$}; Also we will have $a_{z+1}=\sum_{i=\layer}^z \excess-c(\S_i)$.

To prove the claim for any $j$, we will show that $\opt\in \Omega(\min\{\frac{c(\online_j)}{\rho(\Delta_-,\cdot)},\frac{c(\online_j)}{g_{\alpha,\gamma}(|\hat{X}|)}\}+a_j)$. 
    
    We first show that $\opt \ge a_j$. As $a_1=0$, the claim is trivial for $j=1$. If $j\ge2$ the cost that the online algorithm spends at the stage which led to purchasing $\S_{j-1}$ is at most $a_j$, and therefore $a_j$ is a lower bound on $\opt$; This is true as without loss of generality, $\alg$ never pays more than $\opt$ for any element. 
    
    To see this, note that the cheapest set covering any arriving element has cost at most $\opt$. We can always buy such a set, and pay the same amount towards buying the set proposed by the online algorithm. The overall cost is at most twice that of the online algorithm. Alternatively, for any online algorithm, a guess on (a $2$-approximation of) $\opt$ can be maintained online (again losing a factor $2$ in the final guarantee). Hence, all sets of cost larger than this guess need not be considered.

    To complete the argument, we will show $\opt\in \Omega(\frac{c(\online_j)}{\max\{\rho(\Delta_-,\cdot),g_{\alpha,\gamma}(|\hat{X}|)\}})$.  Note that the number of elements handled by our algorithm in layer $j$ is at most $|X_j|$, as otherwise by Property (D), the cost to cover only these elements is at least \smash{$\opt \ge c(\online_j)/g_{\alpha,\gamma}(|\hat{X}|)$}.
    
    So  by the assumption on the competitive ratio of $\alg$ we have \smash{$\opt\in \Omega(\frac{c(\online_j)}{\rho( |X_{j}|,\cdot)})$} (the online algorithm for elements that arrive at layer $j$, is an instance with at most $|X_j|$ elements and hence is $\rho(|X_j|,\cdot)$-competitive). This means that if $\Delta_-\ge  |X_j|/2$ we are done as  $\rho(\Delta_-,\cdot) \in  \Omega(\rho(|X_{j}|,\cdot))$ (note that the competitive ratio can grow at most linearly). So we assume that \smash{$\Delta_-< |X_j|/2$}. As $ X_j\subseteq R_{j-1}$, then $\Delta_-\le |R_{j-1}|$ and by Lemma~\ref{Lem:LinearBoundOnOptGen} if $j>1$, we have \smash{$\opt \ge c(\S_{j-1})/g_{\alpha,\gamma}(|\hat{X}|)$}.

    We distinguish two cases; If $j>1$ and $c(\S_{j})\le 10 c(\S_{j-1})$, then $\opt \ge  c(\S_{j-1})/g_{\alpha,\gamma}(|\hat{X}|)$ already implies $\opt \in \Omega(c(\online_j)/g_{\alpha,\gamma}(|\hat{X}|))$, and hence the claim. Now consider the case $j=1$ and $c(\S_{j})> 10 c(\S_{j-1})$; As \smash{$\Delta_-<|X_j|/2$}, we have \smash{$\Delta_-<|R_{j-1}|/2$}.  By Property (C) for $j>1$ and Property (D) for $j=1$, it holds that \smash{$\opt \ge c(\S_{j})/g_{\alpha,\gamma}(|\hat{X}|)$}.
\end{proof}
\end{fullversion}

In the next claim, we give an upper-bound on the cost of $\alg_-$.

\begin{claim}\label{claim:UpperboundOnCostGen}
    The cost of $\alg_-$ can be upper bounded by $O(c(\S_{i-1})+c(\S_{i})) + \excess.$
\end{claim}

\begin{proof}
    The cost of $\alg_-$ is $O(\sum_{j=1}^i c(\S_j))+\excess.$ By Property~(B), for every $j$, $c(\S_j)+c(\S_{j+1})\le \frac{3}{4}( c(\S_{j+2})+c(\S_{j+3})).$ So \smash{$\sum_{j=1}^{i-2} c(\S_j)\in O(c(\S_{i-1})+c(\S_i))$} and hence the claim.
\end{proof}

Before proving our main result of this section, namely Theorem~\ref{thm:AnalysisGen}, we introduce some useful notation. We define \smash{$k_-:|X\cap \hat X|$} and \smash{$k_+:= |X\setminus \hat{X}|$} as the number of arriving elements that are handled by $\alg_-$ and $\alg_+$, respectively ($k_-+k_+=k$). We first analyze the competitive ratio of $\alg_-$. 

\begin{lemma}\label{lem:AnalysisGenALG-}
    Our algorithm for $ALG_-$ is $O(\rho(\eta,\cdot)+g_{\alpha,\gamma}(k_-))$-competitive. 
\end{lemma}
\begin{proof}
    We first assume \smash{$k_-\ge |\frac{\hat{X}}{2}|$}.
    From Lemma~\ref{Lem:LogarithmicBoundOnOptGen} and Claim~\ref{claim:UpperboundOnCostGen}, we observe that our algorithm is \smash{$O(\rho(\Delta_-,\cdot)+g_{\alpha,\gamma}(|\hat{X}|))$}-competitive. As \smash{$k_-\ge\frac{|\hat{X}|}{2}$}, and by definition of $\Delta_-$, $\eta\ge \min\{k_-,\Delta_-\}$, and therefore \smash{$\rho(\eta,\cdot)+g_{\alpha,\gamma}(|\hat{X}|)\in O(\rho(\eta,\cdot)+g_{\alpha,\gamma}(k_-))$}.

    Now assume $k_-< |\hat{X}|/2$. In this case $k_-<\eta$, and it suffices to show that our algorithm is $O(\rho(k_-,.))$-competitive. For this purpose, we observe that at each layer $j$ of $\alg_-$ our cost is at most $O(\rho(k_-,.))\opt$, as at most $k_-$ elements arrive. Now by Claim~\ref{claim:UpperboundOnCostGen} our algorithm is $O(\rho(k_-,.))$-competitive.
\end{proof}

Now we have all the ingredients to analyze the competitive ratio of Algorithm~\ref{alg:ice}.

\begin{theorem}\label{thm:AnalysisGen}
    Algorithm~\ref{alg:ice} is $O(\rho(\eta,\cdot)+g_{\alpha,\gamma}(k))$-competitive. 
\end{theorem}
\begin{proof}[Proof of Theorem~\ref{thm:AnalysisGen}]
$\alg_+$ is $O(\rho(k_+,\cdot))$-competitive as it is an instance of $\alg$ with $k_+$ arriving elements. By Lemma~\ref{lem:AnalysisGenALG-}, $\alg_-$ is $O(\rho(\eta,\cdot)+g_{\alpha,\gamma}(k_-))$-competitive. By definition of $k_+$, it holds that $k_+\le \eta$ and hence $\rho(k_+,\cdot)\in O(\rho(\eta,\cdot))$. Therefore, Algorithm~\ref{alg:ice} is 
$O(\rho(\min\{\eta,k\},\cdot)+g_{\alpha,\gamma}(k))$-competitive.
\end{proof}
Now the 
proof of Theorems~\ref{thm:exp_variant} and~\ref{thm:poly_variant} follow directly from Theorems~\ref{thm:AnalysisGen} and Theorem~\ref{thm:DecompostionConstruction}.

\section{Applications}\label{sec:applications_formal}
\subsection{Set Cover and Non-Metric Facility Location}\label{sec:set_cover_decomp}

The main theorem of this section is the following.
\begin{theorem}\label{thm:SCwithPrediction}
    There is a polynomial-time, $O(\log\eta\log|\S| + \log k)$- (resp. $O(\log\eta\log|F| + \log k)$-) competitive algorithm for online set cover (resp. non-metric facility location). In exponential time, the dependence on $\log k$ can be removed.
\end{theorem}

    By Theorem~\ref{thm:DecompostionConstruction} combined with Theorem~\ref{thm:poly_variant}, is is sufficient to show that these problems are efficiently $(\alpha, \gamma)$-decomposable. This is shown in Lemma~\ref{lem:kSetCover}. For the exponential-time variant, we use Theorem~\ref{thm:exp_variant} (observe once again that any covering problem is $(1,1)$-decomposable).

\begin{lemma}\label{lem:kSetCover}
    Set cover and Non-metric Facility Location are efficiently $(1, e/(e-1))$-decomposable. 
\end{lemma}

\begin{shortversion}

\begin{proof}[Proof of Lemma~\ref{lem:kSetCover}]
    By binary search, we may assume we know a number $C$, such that the optimal cost of covering $i$ elements of $(X, \S)$ equals~$C$. Convert $(X, \S)$ into an instance of so-called \emph{Budgeted Maximum Coverage Problem (BMC)}: In the BMC problem, given a collection $\mathcal{S}$ of sets defined over a universe of weighted elements, where each set $S \in \mathcal{S}$ has an associated cost $c(S)$, and given a budget $C$, the objective is to select a subcollection $\mathcal{S}' \subseteq \mathcal{S}$ such that the total cost satisfies $\sum_{S \in \mathcal{S}'} c(S) \leq C$, and the total weight of the elements covered by $\mathcal{S}'$ is maximized. Our set cover instance $(X, \S)$ converts into a BMC instance with budget $C$, where each element has unit weight. Khuller et al.~\cite{SMN99} presented a $(1-1/e)$-approximation algorithm for BMC, hence applying it we can cover at least $i\cdot\frac{e-1}{e}$ elements while not exceeding the budget $C$.

By the results from~\cite{Hochbaum, SMN99, Vygen}, similar arguments can be applied for NMFL (see full version). 
\end{proof}
\end{shortversion}

\begin{fullversion}
    
\begin{proof}[Proof of Lemma~\ref{lem:kSetCover}]
    Since NMFL is a generalisation of set cover, we only give the proof for NMFL. Through binary search, we may assume we know a budget $B$ such that the cheapest way to cover $i$ clients from the instance of NMFL has cost at most $B$. We implicitly convert our instance of NMFL into an instance of set cover (SC) with an \emph{exponential} number of sets: for each facility $f$ and each subset of clients $f$ connects to, introduce a set of total cost equal to the facility cost of $f$ plus the respective connection costs. Clearly, any partial solution to NMFL of cost at most $B$ covering at least $i$ clients corresponds to a partial solution to SC of equal cost and covering the same number of elements, and vice versa. 
    
    Now, we can convert this set cover instance into an instance
    of so-called \emph{Budgeted Maximum Coverage Problem (BMC)}: In the BMC problem, given a collection $\mathcal{S}$ of sets defined over a universe of weighted elements, where each set $S \in \mathcal{S}$ has an associated cost $c(S)$, and given a budget $C$, the objective is to select a subcollection $\mathcal{S}' \subseteq \mathcal{S}$ such that the total cost satisfies $\sum_{S \in \mathcal{S}'} c(S) \leq C$, and the total weight of the elements covered by $\mathcal{S}'$ is maximized.  Khuller et al.~\cite{SMN99} presented a $(1-1/e)$-approximation algorithm for BMC.
    
    Our set cover instance SC naturally converts into a BMC instance with budget $B$, and we can use the algorithm of~\cite{SMN99}. Crucially, their algorithm is greedy: in every stage, it considers the most cost-effective set, i.e.\,the set that covers the most uncovered elements per unit cost. In terms of the original instance of NMFL, this set corresponds to the following. Let $\mathsf{Unc}\subseteq C$ be the uncovered clients (elements), $D(f) \subseteq 2^{N(f)}$ a subset of the clients facility $f$ connects to $(N(f) \subseteq C)$, and denote by
    $$ \max_{(f,D(f))}\frac{\cost(f) + \sum_{c \in \mathsf{Unc}\cap D(f)} \cost(c,f)}{|D(f) \cap \mathsf{Unc}|}$$
    the \emph{effectiveness} of facility $f$. The facility-clients pair $(f,D)$ maximizing the effectiveness correspond to the most cost-effective set in SC, and vice versa. To see that this can be found efficiently, for every facility $f$, we order the respective clients it connects to by increasing order of connection cost, $c_f^1, c_f^2, \ldots$. For a given $f$ and given the cardinality of an optimal $D(f)$, it is clearly sufficient to only consider the first $|D(f)|$ clients in increasing order of connection cost. Hence this set can be found by comparing the effectiveness of $|F|\cdot|\mathsf{Unc}| \leq |F|\cdot|C|$ facility-clients pairs.
\end{proof}

\end{fullversion}

Note that this is asymptotically optimal, a $(1, e/(e-1)-\epsilon))$-decomposition in polynomial time implies a polynomial-time $(1-O(\epsilon))\ln|X|$-approximation to set cover, contradicting $\mathsf{P} \neq \mathsf{NP}$,~\cite{DinurSteurer}. Extending this, we can show a corresponding lower bound for online set cover with predictions.

\begin{theorem}\label{thm:LBSetCover}
    No polynomial-time learning-augmented online algorithm for set cover achieves an $O(\log\eta\allowbreak \log|\S|)$-competitive ratio (or $o(\log k)$).
\end{theorem}

\begin{proof}
    Consider an offline instance of set cover $(X, \S)$. We use such a learning-augmented online algorithm to solve it (present the elements of $X$ in arbitrary order). This results in an approximation guarantee of $O(\log|\S|)$. Since set cover does not admit a $O(\log|\S|)$- or a $(1-\epsilon)\cdot\ln|X|$-approximation in polynomial time (assuming $\mathsf{P}\neq\mathsf{NP}$), see~\cite{Nelson07, DinurSteurer}, the proof follows.
\end{proof}

\subsection{Weighted Path Augmentation and Parking Permit Problem}\label{sec:wpap}

The main results of this section is resumed in the following two theorems.

\begin{theorem}\label{thm:main_deterministic_wpap}
    There is a deterministic, $O(\log \eta)$- and a randomized, $O(\log\log\eta)$-competitive learning-augmented online algorithm for weighted path augmentation and parking permit. 
\end{theorem}
For both these results, we combine a $O(\log k)$-competitive (resp. $O(\log \log k)$-competitive) online algorithm for $\wpap$ with our framework. By Theorem~\ref{thm:DecompostionConstruction} and Theorem~\ref{thm:poly_variant}, is is sufficient to show that the problem is efficiently (1,1)-decomposable.

\begin{lemma}\label{lem:DP}
    $\wpap$ is efficiently $(1,1)$-decomposable.
\end{lemma}

\begin{shortversion}
    \begin{proof}[Sketch of proof of Lemma~\ref{lem:DP} (see full version)]
    We proceed by dynamic programming. We consider $\DP(\ell, [a, b])$, which denotes the cheapest cost to cover exactly $\ell$ elements among the (consecutive) elements labeled from $a$ to $b$. The base-case $\DP(b-a+1, [a,b])$ can be computed exactly, since the problem can be formulated as an LP with a totally unimodular constraint matrix~\cite[Chapter~$19$]{Schrijver_linearandintegerprogramming}. 
    \end{proof}
\end{shortversion}
    
\begin{fullversion}
\begin{proof}[Proof of Lemma~\ref{lem:DP}]
    Consider an instance of $\wpap$ with  elements $E$, where the elements are indexed from $1$ to $n$. For every $1\le i\le n$ and in $\text{poly}(n)$-time, we show how to compute a subset $A \subseteq E$ of size at least $i$, and a set of links in $L \subseteq \L$ that covers $A$, such that $\cost(L)$ is minimized.
    We handle this by dynamic programming. We define $\DP[i, j]$ to be the minimum cost to cover $i$ items from $j$ to $n$, only allowing links that do not contain elements from $1$ to $j-1$. If this is infeasible, set the value to $+\infty$. Clearly, when $i = n-j+1$, this can be solved exactly, since the constraint matrix is totally unimodular, since it has the consecutive ones property~\cite[Chapter~$19$]{Schrijver_linearandintegerprogramming}.
    To update the table for other values, we proceed as follows. Denote by $c_{j,k}$ the cheapest cost of covering all elements between $j$ and $k$ (including the endpoints) with sets not including any element after~$k$ (if $k < j$, this is $0$). Again, $c_{j,k}$ can be computed exactly by total unimodularity. We set
    $$\DP[i+1,j] := \min_{j \leq k \leq j+i} c_{j,k-1} + \DP[i-(k-j), k+1].$$
    Correctness follows from the observation that if the $k^{th}$ element $(i \leq k \leq n)$ is the first element not to be selected for a potential solution in $\DP[i+1,j]$, the cost of covering elements $i$ to $k-1$ is exactly $c_{j, k-1}$, and, the cost of selecting the remaining $i-(k-j)$ elements equals $\DP[i-(k-j), k+1]$, since any link starting earlier than $k+1$ includes $k$ and can be disregarded. The running time of this algorithm is $\mathsf{poly}(n)$ as $i$, $j$ and $k$ are at most $n$.
    \end{proof}
\end{fullversion}

The authors in~\cite{Meyerson05ParkingPermit, NaorUmbohWilliamson} give a 
deterministic $O(K)$-competitive algorithm for $\wpap$, where $K$ is the number of different cost values for the links. $K$ can be assumed to be of order $O(\log |E|)$. Furthermore, in~\cite{Meyerson05ParkingPermit}, they provide a \emph{fractional}, $O(\log K)$-competitive algorithm for parking permit. In this setting, one has to maintain a monotonically increasing, fractional solution covering all arriving elements, i.e.\,the sum of the fractional values assigned to links covering each arrived element needs to sum up to at least one. To convert such a fractional solution into an integral solution, they provide a randomized \emph{online rounding scheme} which only loses a constant factor (with respect to the fractional solution).
That is, they obtain a 
randomized $O(\log K)$-competitive algorithm for parking permit.

To use these results, we need to replace the dependence on $|E|$ with a dependence on $k$. To do so, we rely on the following technical lemma, which holds \emph{irrespective} of the knowledge of $k$.

\begin{lemma}\label{lem:laminarity_wpap}
    Consider an instance $(E, \L, \cost)$ of $\text{ }\wpap$. Up to losing a constant factor, one can assume that the links in $\L$ are \emph{laminar} (i.e., the elements covered by any two links are either disjoint or contained in one another). Furthermore, to get a $O(\log k)$ (resp. $O(\log \log k)$) competitive algorithm, one can assume that there are at most $O(\log k)$ (resp. \smash{$O(\log^2 k)$}) different cost classes.
\end{lemma}

\begin{fullversion}
\begin{proof}[Proof of Lemma~\ref{lem:laminarity_wpap}]
    Note that a guess on $\opt$ can be maintained online, see e.g.~\cite{AlonAwerbuchBuchbinderNaor}. We now show that we can restrict to $O(\log^2 k)$ different cost classes in the setting where we use a $O(\log\log k)$-competitive for an instance of $\wpap$, irrespective of the knowledge of $k$. To this end, we maintain a $2$-approximation $\mathsf{guess}$ on $\log\log k$, that is
    $$\mathsf{guess} \leq \log\log k < 2\cdot \mathsf{guess},$$
    where $k$ denotes the number of elements seen so far. Whenever this exceeds $2\cdot\mathsf{guess}$, we double our guess, i.e.\,$\mathsf{guess} \leftarrow 2\cdot\mathsf{guess}$, and we re-run the algorithm from scratch. In each such run with guess $\mathsf{guess}_i$, the cost of the online algorithm is at most $\mathsf{guess}_i\cdot\opt$. Since the $\{\mathsf{guess}_i\}_{i\geq 1}$ are geometrically increasing, this is dominated by the last run (which uses the correct estimate for $\log\log k$), hence the total cost incurred is at most $O(\log\log k)\cdot \opt$. It remains to show that we may restrict to at most \smash{$O(\log^2 k)$} classes of links. To this end, set 
    $$C_i = 2^{2^{2\cdot \mathsf{guess}_i}},$$
    for each run. We claim that we may restrict to links with cost between $\opt$ and $\opt/C_i$. Indeed, all links of cost larger than $\opt$ can be disregarded, and all links of $\opt/C_i$ can be greedily chosen whenever an uncovered elements arrives, incurring total cost at most $\opt$. This is since $C_i$ is strictly larger than the total number of arriving elements in the $i^{\text{th}}$ run. Hence, by rounding the cost of a links to the next power of $4$, we may assume we have at most \smash{$\log(C_i) \leq \log^2 k$} different cost classes of links. The setting where we use a $O(\log k)$-competitive algorithm can be proved in an analogous way, though the number of different cost classes can then be restricted to $O(\log k)$.
    
   We now proceed to show laminarity. For each cost class, we may remove links that are contained in other links, and short-cut the remaining ones as to make them disjoint. This loses a factor of $2$ in the final approximation. Finally, to make the instance laminar, we proceed in increasing order of cost classes, from left to right. For any cost class, we start with the leftmost link, and add to it the (at most) two links from the lower cost class that intersect its left and right boundary. We then shorten the next link in that cost class to make it disjoint from the previous link and proceed with it. Note that this cost of these new links have increased by at most a factor of $2$: the cost of any resulting link at layer $j$ is at most that of the original link at that layer, plus the cost of two original links at layers $j-1, j-2, \ldots, 0$. Since the cost classes are geometrically decreasing, the claim follows. Finally, it remains to show that the resulting set of links is laminar. That is, for any two links $\ell_1, \ell_2$, if $\ell_1 \cap \ell_2 \neq \emptyset$, then either $\ell_1 \subseteq \ell_2$, or $\ell_2 \subseteq \ell_1$. To show this, we proceed by induction. For the basecase, i.e.\,the links of smallest cost class, this is immediate. Hence, assume all resulting links of cost class $j-1$ and lower are laminar. Once we start processing the leftmost link at layer $j$, we add to it the two resulting links crossing its left and right boundary from layer $j-1$. Since both of these links are laminar with respect to links of cost classes smaller than $j-1$, and their union is a consecutive interval, the resulting link contains all links from lower levels that it intersects. Since we shorten the next link on layer $j$ to make it disjoint from this newly created link, the next link on layer $j$ does not intersect any links the previous link contains. All in all, this results in a laminar set of links, with cost that is at most $O(1)$ times higher.
   
\end{proof}
\end{fullversion}

\begin{lemma}[implicit in~\cite{Meyerson05ParkingPermit}]\label{lem:rounding_wpap}
    Let $(E, \L, \cost)$ be an instance with a laminar set of links $\L$. Then, for any monotonically increasing fractional solution of cost $\cost_{\mathsf{frac}}$, there exists an online rounding scheme with expected cost at most $O(1)\cdot \cost_{\mathsf{frac}}$.
\end{lemma}

\begin{fullversion}
    \begin{proof}[Proof of Lemma~\ref{lem:rounding_wpap}]
    Up to rounding to the next power of $2$, we may assume that each link has weight being a power of $2$. This loses a factor of $2$ in the final approximation guarantee. For all such resulting weight classes of links, we remove the links which are entirely contained in another link. It follows then that each
    element $e$ is covered by at most one link in each layer. Denote by~$x$ the fractional, monotonically increasing solution to the instance of $\wpap$. Since the elements arrive one-by-one and each element $e$ is covered by at most one link in each layer in $L := \{0, 1, \ldots \}$ (each of cost $2^{i}$), we index this monotonically increasing solution by $x_i(e)$, where $i\in L$. For the rounding scheme, we proceed as follows. Before the arrival of the first element, we draw a number $\tau\sim[0,1]$ uniformly at random. Whenever element $e$ arrives, we update the solution $x$, and we find $i \in L$ such that
    $$\sum_{j \geq i+1}^{}x_j(e) < \tau, \text{  and,  } \sum_{j \geq i}^{}x_j(e) \geq \tau.$$
    We then buy the unique link at level~$i$ covering $e$. We claim that the expected cost (over the randomness of $\tau$) is proportional to the cost of the fractional solution. To see this, consider some link $\ell$ at layer~$i$. Denote by $e_{\mathsf{first}}^{\ell}$ the first element to appear that lies inside $\ell$, and denote by $e_{\mathsf{final}}^{\ell}$ the last element to appear that is covered by $\ell$. Note that this is well defined, since the adversary is oblivious, i.e.\,the input sequence is fixed. To buy link $\ell$ at some element $e \in \{e_\mathsf{first}^\ell, \ldots, e_{\mathsf{last}}^\ell\}$, it must hold that
    $$\sum_{j \geq i+1}^{}x_j(e_\mathsf{first}^\ell) \stackrel{(*)}{\leq} \sum_{j \geq i+1}^{}x_j(e) < \tau, \text{  and, }$$
    $$ \tau \leq \sum_{j \geq i}^{}x_j(e) \stackrel{(*)}{\leq} \sum_{j \geq i}^{}x_j(e_{\mathsf{last}}^\ell). $$
    Here, $(*)$ follows from laminarity and monotonicity. Hence, the probability of purchasing link $\ell$ at level~$i$ is at most 
    $$ \sum_{j \geq i}^{}x_j(e_{\mathsf{last}}^\ell) -  \sum_{j \geq i+1}^{}x_j(e_\mathsf{first}^\ell).$$
    Observing that the fractional value of a link only changes if the arriving element is contained in it, and denoting $\Delta x_j(e)$ as the \emph{increase} in $x_j(\cdot)$ upon the arrival of element $e$, we can rewrite this as
    $$ x_i(e_\mathsf{last}^\ell) + \sum_{j \geq i+1} \sum_{e \in \ell} \Delta x_j(e).$$ 
    The total expected cost is thus upper bounded by
    $$ \sum_{i \in L} \sum_{\{\ell \mid \cost(\ell) = 2^i\}} 2^i \left(x_i(e_\mathsf{last}^\ell) + \sum_{j \geq i+1}\sum_{e \in \ell} \Delta x_j(e)\right). $$
    Since each element is contained in at most one link per level, we can rewrite this as
    \begin{align*}
        \sum_{i \in L} \sum_{\{\ell \mid \cost(\ell) = 2^i\}} 2^i x_i(e_\mathsf{last}^\ell) + \sum_{j \in L} \sum_{\{\ell \mid \cost(\ell) = 2^j\}}\sum_{0 \leq i < j}2^{j-i}\underbrace{\sum_{e \in \ell}  \Delta x_j (e)}_{\leq x_j(e^\ell_{\mathsf{last}})} \\
        \leq O(1)\cdot \sum_{i \in L} \sum_{\{\ell \mid \cost(\ell) = 2^i\}} 2^i \cdot x_i(e_\mathsf{last}^\ell).
    \end{align*}
    The right-hand-side is the cost of the fractional solution $x_i(\cdot)$, up to a constant factor. The lemma follows.
    
\end{proof}
\end{fullversion}

\begin{theorem}\label{thm:k_comp_algorithm_wpap}
    There is a deterministic, $O(\log k)$-competitive algorithm for $\wpap$, as well as a randomized $O(\log\log k)$-competitive algorithm.
    \end{theorem}
    
    \begin{proof}[Proof of Theorem~\ref{thm:k_comp_algorithm_wpap}]
        For the deterministic algorithm, by Lemma~\ref{lem:laminarity_wpap}, we may assume that the instance of $\wpap$ is laminar and that there are at most $O(\log k)$ different cost classes. Since the deterministic algorithms of~\cite{Meyerson05ParkingPermit, NaorUmbohWilliamson} are $O(K)$-competitive with respect to the number of different cost classes, the claim follows. For the randomized algorithm, we proceed by randomized rounding. By the results in~\cite{BuchbinderNaor_online_primal_dual}, it is possible to maintain a monotonically increasing fractional solution to an instance of set cover, that is $O(\log d)$-competitive, where $d$ is the maximum number of sets each element is contained in. In our setting, links correspond to sets and thus \smash{$d = O(\log^2 k)$}. By Lemma~\ref{lem:rounding_wpap}, we can round this $O(\log\log k)$-competitive solution in an online fashion, incurring only constant factor loss. This yields the $O(\log\log k)$-competitive algorithm.
    \end{proof}

\begin{fullversion}

\begin{theorem}\label{thm:lower_bound_if_links_not_known}
    If $(E, \L, \cost)$ is not given in advance, there is a lower bound of $\Omega(k)$ on the competitive ratio of any deterministic or randomized online algorithm for $\wpap$.
\end{theorem}

\begin{proof}[Proof of Theorem~\ref{thm:lower_bound_if_links_not_known}]
    The lower bound instance looks the same for both the deterministic and randomized setting: Each link has unit cost. The first link only covers the first element, the second link only covers the first and second element, the third covers the first, second and third element, and so on. 
    
    In the deterministic setting, the adversary releases the elements one-by-one, in order. Up to reordering the links, among all eligible links, we may assume that the online algorithm picks the one of smallest index. This yields the competitive ratio of $k = |E|$. 

    In the randomized setting, the adversary releases $k = \log |E|$ elements indexed by $\lfloor (1-2^{-j})\cdot|E| \rfloor$, for $j \in \{1, 2, \ldots, \log |E|\}$, in that order. Note that whenever an uncovered element arrives, among all links covering it, the randomized online algorithm will select one such link uniformly at random. This is since all of the links remaining up to this point "look the same". We may also assume without loss of generality, that the online algorithm only buys a single link if an uncovered element appears. Hence, the probability that the \smash{$\text{j}^{\text{th}}$}-element is not covered by a link covering elements $1, \ldots, j-1$ can be lower bounded by
    $$(1-2^{-j})\cdot(2^{-1}-2^{-j})\cdot \, \cdots \, \cdot (2^{-j+1}-2^{-j}).$$
    Rearranging and simplifying, we can further lower bound the probability of the online algorithm having to add a link to cover $e$ by 
    $$\prod_{i \geq 1} (1-2^{-i}).$$
    Taking the logarithm, using that $\log(1-x) \geq 2\log(2)\cdot x$ and exponentiating, it follows that this is at least $1/4$. Hence, the lower bound of $\Omega(k)$ follows.
\end{proof}

\end{fullversion}

\subsection{Metric facility Location and Steiner variants}\label{sec:steiner}
The main result from this section is the following.

\begin{theorem}\label{thm:metric_facility_location}
    There is a polynomial-time, $O(\tfrac{\log\eta}{\log\log\eta})$-competitive learning-augmented algorithm for metric facility location.
\end{theorem}
 
\begin{proof}[Proof of Theorem~\ref{thm:metric_facility_location}]
    By the results of~\cite{Meyerson_online_facility_location, Fotakis_facility_location}, there is a deterministic \smash{$O(\tfrac{\log k}{\log \log k})$}-competitive algorithm for facility location. Hence, this result follows in a black-box way from our framework, provided the problem is efficiently $(O(1), 1)$-decomposable. Decomposability follows from a problem that has been studied in the setting of \emph{facility location with outliers}. In particular, in~\cite[Theorem~$5.2$]{CharikarKhullerMountNarasimhan}, the authors give a $3$-approximation to problem of serving a prescribed fraction of the clients as cheaply as possible, meaning that a $O(3, 1)$-decomposition can be computed in polynomial-time.
\end{proof}
We now list other prominent problems in network design where our framework applies. 
\begin{fullversion}
    With the exception of (edge-weighted) Steiner tree, we do not know how to turn our approach polynomial, i.e. how to compute an $(O(1), 1)$-decomposition in polynomial time. Note that most of these problems generalise the set cover problem, hence this is impossible in polynomial-time. 
\end{fullversion}
\begin{theorem}[informal]
    For all the problems listed below, our framework provides a learning-augmented algorithm with competitive ratio $O(\rho(\eta, \cdot))$.
\end{theorem}

\begin{shortversion}

\begin{itemize}
    \item Steiner tree and Steiner forest: These problems have $O(\log k)$-competitive online algorithms,~\cite{ImaseWaxman, BermanCoulston} and are the only two Steiner variants considered previously in the literature in our learning-augmented setting~\cite{MoseleyXu, AzarPanigrahiTouitouOnlineGraphAlgo}. Results were discussed in the introduction. 
    \item Connected facility location: This is a variation of the facility location problem, where the facilities need to remain connected. As such, it generalises the online Steiner tree problem. The best competitive ratio is $O(\log k)$,~\cite{Umboh}.
    \item Group Steiner Tree and Forest: 
    Groups of terminals (terminal pairs) arrive one-by-one, and for each group, one needs to select one terminal (pair) to be included in the Steiner Tree (forest). For edge-weighted graphs, the best competitive ratio is of order \smash{$O(\log^7 k \log^5 |V|)$},~\cite{NaorPanigrahiSingh}.
    \item Node-weighted Steiner problems: 
    Defined analogously as in (edge weighted) Steiner problems, but the costs are on vertices (and edges). To use an edge, one needs to also have purchased its two incident vertices. The best achievable competitive ratio for Steiner problems is typically a poly-logarithmic factor higher than the edge-weighted setting,~\cite{NaorPanigrahiSingh, BorstEliasVenzin25}, and for generalizations such as the group-setting, the best-known running time is \emph{quasi-polynomial}, i.e.\,of the form \smash{$|V|^{\mathsf{poly}(\log |V|)}$}.
    \item Directed Steiner problems: This further generalizes the node- and edge-weighted setting. In polynomial-time, there is a \smash{$O(k^{\epsilon}\log^{O(1)}(|V|))$}- (\smash{$O(k^{1/2 + \epsilon}\log^{O(1)}(|V|))$}-) competitive algorithm for directed Steiner tree (forest), for any $\epsilon > 0$,~\cite{ChakrabartyEneKrishnaswamyPanigrahi}.

\end{itemize}

\end{shortversion}

\begin{fullversion}

\begin{itemize}
    \item Steiner tree and Steiner forest: These problems have $O(\log k)$-competitive online algorithms,~\cite{ImaseWaxman, BermanCoulston}. 
    These are the only two Steiner variants considered previously in the literature in our learning-augmented setting. Results were discussed in the introduction (\cite{MoseleyXu, AzarPanigrahiTouitouOnlineGraphAlgo}). 
    \item Connected facility location: This is a variation of the facility location problem, where the facilities need to remain connected. As such, it generalises the online Steiner tree problem. The best competitive ratio is $O(\log k)$,~\cite{Umboh}.
    \item Group Steiner Tree and Forest: The group setting is an important generalization of Steiner problems. Groups of terminals (terminal pairs) arrive one-by-one, and for each group, one needs to select one terminal (pair) to be included in the Steiner Tree (forest). For edge-weighted graphs, the best competitive ratio is of order \smash{$O(\log^7 k \log^5 |V|)$},~\cite{NaorPanigrahiSingh}.
    \item Node-weighted Steiner problems: This generalizes (edge-weighted) Steiner problems and set cover problems. Problems are defined analogously, but the costs are on vertices (and edges). To use an edge, one needs to also have purchased its two incident vertices. The best achievable competitive ratio for Steiner problems is typically a poly-logarithmic factor higher than the edge-weighted setting,~\cite{NaorPanigrahiSingh, BorstEliasVenzin25}, and for generalizations such as the group-setting, the best-known running time is \emph{quasi-polynomial}, i.e.\,of the form \smash{$|V|^{\mathsf{poly}(\log |V|)}$}.
    \item Directed Steiner problems: This further generalizes the node- and edge-weighted setting. In polynomial-time, there is a \smash{$O(k^{\epsilon}\log^{O(1)}(|V|))$}- (\smash{$O(k^{1/2 + \epsilon}\log^{O(1)}(|V|))$}-) competitive algorithm for directed Steiner tree (forest), for any $\epsilon > 0$,~\cite{ChakrabartyEneKrishnaswamyPanigrahi}.

\end{itemize}

In the remainder of the section, we show how to recover the results from~\cite{MoseleyXu} (and partially~\cite{AzarPanigrahiTouitouOnlineGraphAlgo}).

\begin{theorem}\label{thm:learning_augmented_steiner_tree}
    Our framework provides a polynomial-time learning-augmented $O(\log\eta)$-competitive online for online Steiner tree.
\end{theorem}

We make use of the $k$-MST and $k$-Steiner tree problem.

\begin{definition}[$k$-MST Problem]
    In the $k$-Minimum Spanning Tree problem, we are given an edge-weighted graph $G = (V, E)$, and a designated root vertex $r \in V$. The objective is to compute a tree rooted at $r$ that spans at least $k$ vertices while minimizing the total cost.

\end{definition}

There exist several constant-factor approximation algorithms for the $k$-MST problem (see~\cite{AK00, AR98, BRV99}), culminating in a $2$-approximation algorithm presented in~\cite{G2005}.

\begin{definition}[$k$-Steiner Tree Problem]
    In the $k$-Steiner Tree problem, we are given an edge-weighted graph $G = (V, E)$, where the vertex set is partitioned into terminals $\{r\} \cup T$ and Steiner nodes $S$, i.e., $V = T \cup S$. The objective is to compute a tree rooted at $r$ that spans at least $k$ terminals while minimizing the total cost of its edges.
\end{definition}

\begin{lemma}\label{lem:k_Steiner_tree}
    There exists a $4$-approximation for the minimum $k$-Steiner tree problem.
\end{lemma}

\begin{proof}
    We shortcut all Steiner nodes $S$, by passing to the \emph{metric closure} restricted to terminals in $\{r\} \cup T$. That is, we consider the complete graph on vertex set $\{r\} \cup T$, where the cost of an edge between any pair of vertices corresponds to the cost of their shortest path in the original instance. Any tree can be mapped to the original instance with no greater cost. Conversely, by passing to the metric closure, the cost of any tree on terminals is increased by at most a factor of $2$,~\cite{KouMarkowskyBerman}. Hence, using the $2$-approximation for $k$-MST with root $r$ on the metric closure, and mapping back to the original graph, connects at least $k$ terminals with the root, and has total cost at most $4$ times that of a minimum $k$-Steiner tree. 
\end{proof}

\begin{proof}[Proof of Theorem~\ref{thm:learning_augmented_steiner_tree}]
Lemma~\ref{lem:k_Steiner_tree} shows that this problem is efficiently $(4,1)$-decomposable. As there there exists a $O(\log k)$-approximation for the Online Steiner Tree Problem~\cite{IW91}, our framework provides a $O(\log\eta)$-competitive online algorithm for Steiner tree.
\end{proof}

\end{fullversion}

\section{Experimental evaluation}\label{sec:experimental_evaluation}
In this section, we evaluate the performance of our framework for the set cover problem.

\begin{figure}[h]
    \centering
    \begin{subfigure}[b]{0.49\textwidth}
        \centering
        \includegraphics[width=\textwidth]{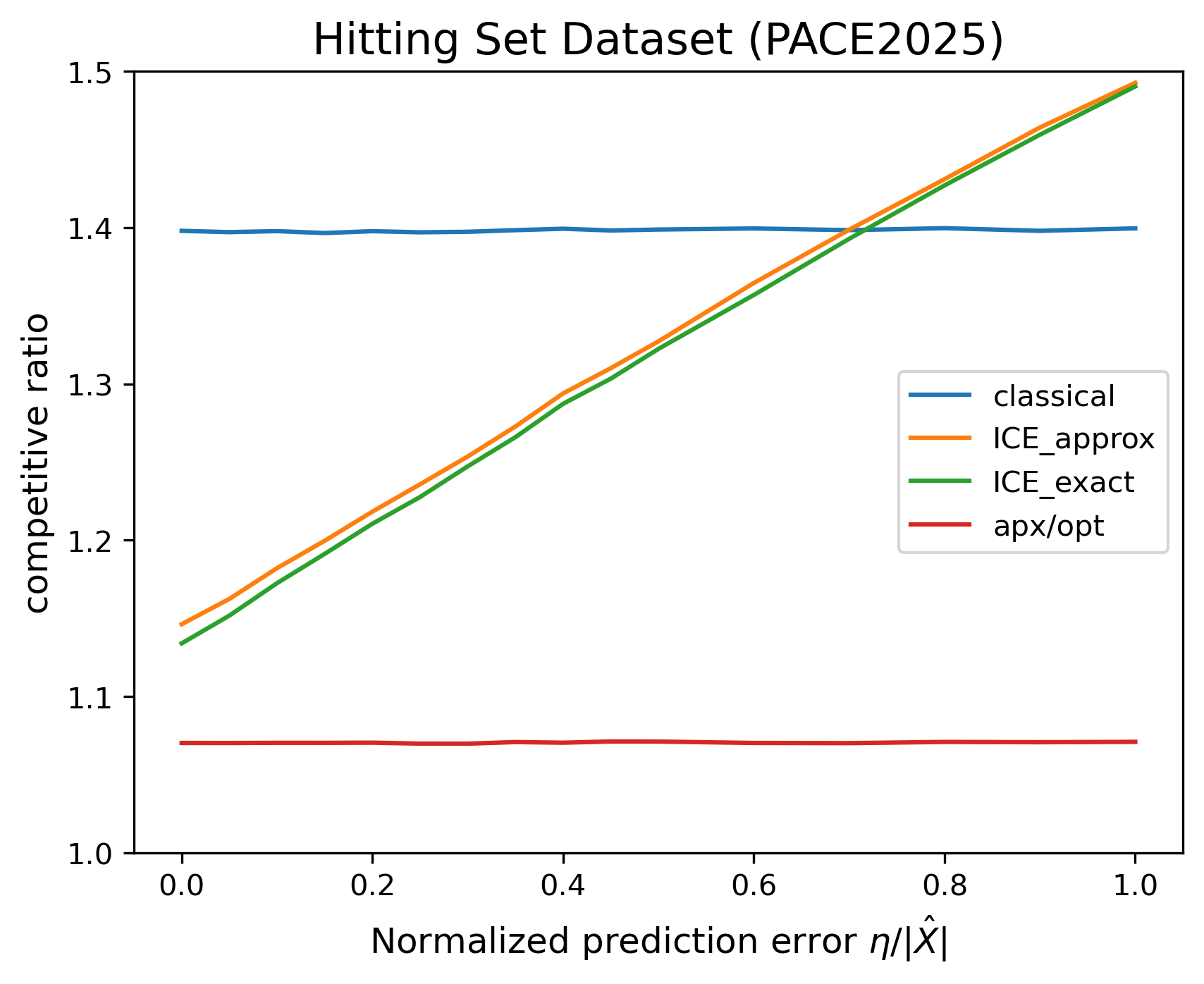}
    \end{subfigure}
    \hfill
    \begin{subfigure}[b]{0.49\textwidth}
        \centering
        \includegraphics[width=\textwidth]{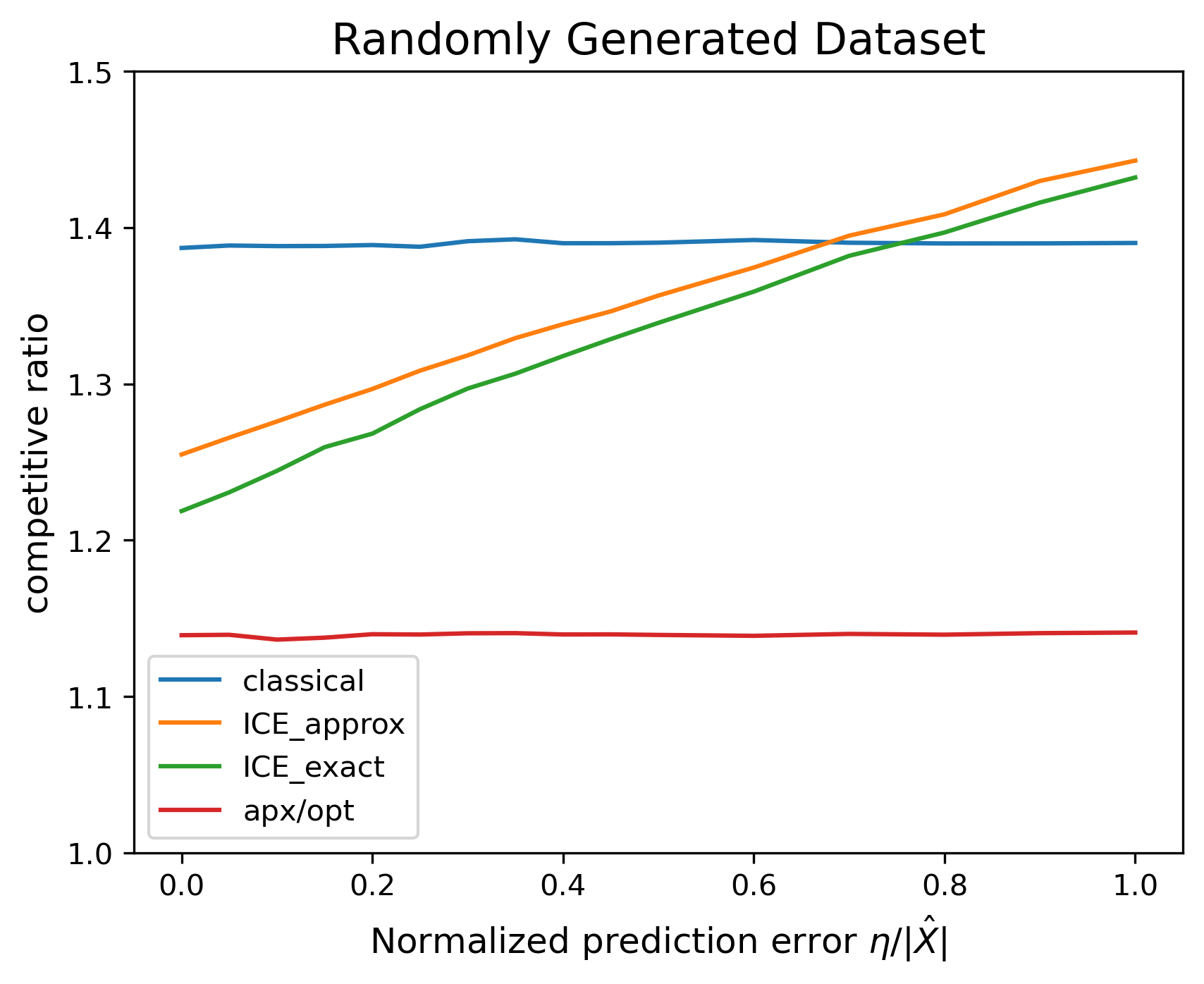}
    \end{subfigure}
\begin{fullversion}
\begin{minipage}{\textwidth}
\tiny
    \centering
        \begin{tabular}{lcccccccc}
\toprule
$\eta$ \% & 0 & 10 & 20 & 30 & 40 & 50 & 60 & 70  \\
\midrule
ICE\_exact (PACE) & 1.13 (0.04) & 1.17 (0.04) & 1.21 (0.05) & 1.25 (0.05) & 1.29 (0.06) & 1.32 (0.07) & 1.36 (0.08) & 1.39 (0.09) \\
ICE\_approx (PACE) & 1.15 (0.04) & 1.18 (0.04) & 1.22 (0.05) & 1.25 (0.06) & 1.29 (0.07) & 1.33 (0.08) & 1.36 (0.09) & 1.4 (0.1)  \\
classical (PACE) & 1.4 (0.06) & 1.4 (0.06) & 1.4 (0.06) & 1.4 (0.06) & 1.4 (0.06) & 1.4 (0.06) & 1.4 (0.06) & 1.4 (0.06) \\
ICE\_exact (random) & 1.22 (0.04) & 1.24 (0.05) & 1.27 (0.05) & 1.3 (0.05) & 1.32 (0.05) & 1.34 (0.06) & 1.36 (0.06) & 1.38 (0.06) \\
ICE\_approx (random) & 1.25 (0.05) & 1.28 (0.05) & 1.3 (0.05) & 1.32 (0.06) & 1.34 (0.05) & 1.36 (0.06) & 1.37 (0.06) & 1.39 (0.06) \\
classical (random) & 1.39 (0.04) & 1.39 (0.04) & 1.39 (0.04) & 1.39 (0.04) & 1.39 (0.04) & 1.39 (0.04) & 1.39 (0.04) & 1.39 (0.04) \\
\bottomrule
\end{tabular}
\end{minipage}
\end{fullversion}
    \caption{Competitive ratio for online set cover for varying prediction error $\eta$.}
    \label{fig:comp_ratio}

\end{figure}
\textbf{Dataset} We use two datasets. The first one is the public dataset from the ongoing PACE Challenge, from the exact track on the Hitting Set Problem,~\cite{PACE2025} (this is equivalent to the Set Cover problem). On average, these instances have $~2200$ elements, and $~1800$ sets. The second dataset is randomly generated. Each instance consists of $1000$ elements and $100$ sets, each set contains $50$ elements sampled uniformly from the groundset without replacement. Both datasets consist of $100$ instances. 

\textbf{Set-up and predictions} For every instance, we randomly fix $50\%$ of its elements to be included in the prediction \smash{$\hat{X}$}. We obtain the real instances $X$ from \smash{$\hat{X}$} by replacing an $\alpha$-fraction of the predicted elements by unpredicted elements, keeping the total number of arriving elements fixed. This results in the normalized prediction error \smash{$\eta / |\hat{X}| \, = 2\cdot \alpha$}. For each obtained instance $X$, we report the multiplicative gap between the solution of the greedy algorithm and that of an IP-solver (apx/opt).

\textbf{Evaluation} We use the the standard $O(\log k \log |\S|)$-competitive algorithm from~\cite{AlonAwerbuchBuchbinderNaor} to instantiate ICE (Algorithm~\ref{alg:ice}), and also use it as our baseline (classical). We compare the baseline to the (averaged) performance of our approach: we either use an exact decomposition for \smash{$\hat{X}$} (ICE\_exact), corresponding to Theorem~\ref{thm:exp_variant}, or an approximate decomposition (ICE\_approx), corresponding to Theorem~\ref{thm:poly_variant}. Whenever the baseline algorithm is about to make an arbitrary choice among several tied sets, we use the decomposition to break the tie. This does not affect our theoretical guarantees in any way. We achieve substantial improvements over the baseline, even for large prediction errors. The effects between using an exact or an approximate decomposition are negligible for the PACE dataset and quite small for the randomly generated instances, see Figure~\ref{fig:comp_ratio}. \begin{shortversion}
    Detailed statistics can be found in the full version and supplementary material.
\end{shortversion}

\begin{fullversion}
\textbf{Implementation} We use the Gurobi MIP-solver to compute the exact decomposition and the optima. They offer free academic licenses for researchers,~\cite{gurobi}. The rest is implemented in Python. The experiments were run on a High Performance Cluster, taking around $100$ hours to complete.
\end{fullversion}

\section{Acknowledgments}
L.S. and M.V. acknowledge support from 
the NWO VIDI grant VI.Vidi.193.087, and from
the FAIR - Future
Artificial Intelligence Research project (FAIR PE00000013 – CUP B43C22000800006).

\bibliographystyle{plainurl}
\bibliography{references_polished}
\end{document}